\pdfoutput =1
  
\documentclass[pra,onecolumn,superscriptaddress]{article}

\usepackage[T1]{fontenc}
\usepackage[utf8]{inputenc}
\usepackage{amsmath,amsfonts,amsthm,amssymb}
\usepackage{mathtools}
\usepackage{subeqnarray}
\usepackage{setspace}
\usepackage{graphics,graphicx,color}
\usepackage{url}
\usepackage{enumerate}
\usepackage{float} 
\usepackage{multirow}
\usepackage{hhline}
\usepackage[margin=1.2in]{geometry}
\usepackage{braket}
\usepackage{dsfont} 
\usepackage{authblk}
\usepackage{multirow}
\usepackage{hhline}
\usepackage[makeroom]{cancel}

\newtheorem{theorem}{Theorem}

\newtheorem{definition}{Definition}

\newtheorem{prop}{Proposition}

\usepackage[scr=boondoxo,scrscaled=1.05]{mathalfa}

\newcommand{\beq}{\begin{eqnarray}}
\newcommand{\eeq}{\end{eqnarray}}

\begin{document}
\title{Smart contracts meet quantum cryptography}

\author{Andrea Coladangelo}

\affil{Computing and Mathematical Sciences,	
Caltech\\ 	
{acoladan@caltech.edu}
}

\date{}

\sloppy

\maketitle

\abstract{We put forward the idea that classical blockchains and smart contracts are potentially useful primitives not only for classical cryptography, but for quantum cryptography as well. Abstractly, a smart contract is a functionality that allows parties to deposit funds, and release them upon fulfillment of algorithmically checkable conditions, and can thus be employed as a formal tool to enforce monetary incentives.
In this work, we give the first example of the use of smart contracts in a quantum setting.

We describe a simple hybrid classical-quantum payment system whose main ingredients are a classical blockchain capable of handling stateful smart contracts, and quantum lightning, a strengthening of public-key quantum money introduced by Zhandry \cite{zhandry2017quantum}. Our hybrid payment system employs quantum states as banknotes and a classical blockchain to settle disputes and to keep track of the valid serial numbers. It has several desirable properties: it is decentralized, requiring no trust in any single entity; payments are as quick as quantum communication, regardless of the total number of users; when a quantum banknote is damaged or lost, the rightful owner can recover the lost value. 
}

\tableofcontents

\newpage


\section{Introduction}
Cryptocurrencies, along with blockchains and smart contracts, have recently risen to popular attention, with the most well-known examples being Bitcoin and Ethereum \cite{bitcoin, ethereum}. This rise in popularity has been accompanied by increasing discussion about the potential for applications of blockchains and smart contracts beyond simple transactions. 

Informally, a blockchain is a public ledger consisting of a sequence of blocks. Each block typically contains information about a set of transactions, and a new block is appended regularly via a consensus mechanism that involves the parties of a network and no a priori trusted authority. A blockchain is endowed with a native currency which is employed in transactions, and whose basic unit is a ``coin''. The simplest type of transaction is a \textit{payment}, which transfers coins from one party to another. However, more general transactions are allowed, which are known as \textit{smart contracts}. These can be thought of as contracts stored on a blockchain, and whose consequences are executed upon fulfilment of algorithmically checkable conditions.

The idea of introducing monetary incentives has led to interesting applications in classical cryptography via the study of rational proofs \cite{azar2012rational, guo2014rational, guo2016rational, chen2016rational, azar2016incentivize}.
Smart contracts on a blockchain make it possible to enforce these monetary incentives in a trust-less way, i.e. without having to trust a central entity, or that the verifier pays the correct rewards to the provers. They have found several applications in the construction of multi-party cryptographic primitives like secure multi-party computation \cite{bentov2014use}, secure lotteries \cite{bentov2014use}, and protocols for playing card-games \cite{bentov2017instantaneous, david2017kaleidoscope, david201821}.

A central issue that needs to be resolved for blockchains to achieve mass-adoption is \textit{scalability}. This refers to the problem of increasing the throughput of transactions (i.e. transactions per second) while maintaining the resources needed for a party to participate in the consensus mechanism approximately constant, and while maintaining security against adversaries that can corrupt constant fractions of the parties in the network. For example, Bitcoin and Ethereum, can currently handle only on the order of $10$ transactions per second (Visa for a comparison handles about $3500$ per second). 

In this work, we show that smart contracts can be combined with tools from quantum cryptography to provide a simple quantum solution to the scalability problem. We clarify that this solution only solves the scalability problem for \textit{payment} transactions, and not for the more general \textit{smart contracts} transactions.

The main quantum ingredient that we employ is a primitive called quantum lightning, introduced by Zhandry \cite{zhandry2017quantum} and inspired by Lutomirski et al.'s notion of collision resistant quantum money \cite{lutomirski2009breaking}. In a public-key quantum money scheme, a bank is entrusted with generating quantum states (we refer to these as quantum banknotes) with an associated serial number, and a public verification procedure allows anyone in possession of the banknote to check its validity. Importantly, trust is placed in the fact that the central bank will not create multiple quantum banknotes with the same serial number. 
A quantum lightning scheme has the additional feature that no generation procedure, not even the honest one (and hence not even a bank!), can produce two valid money states with the same serial number, except with negligible probability. This opens to the possibility of having a completely decentralized quantum money scheme. However, if the system ought to be trust-less, some issues need to be addressed: most importantly, who is allowed to mint money? Who decides which serial numbers are valid? A second kind of issue that one might consider when thinking about a future practical deployment of such a money scheme is that quantum states almost inevitably deteriorate with time. Unless there exists a technology reliable enough to store large states for extended periods of time, quantum money would most likely not be feasible unless there is a mechanism that allows a user to replace deteriorated banknotes with fresh ones, or to recover the value that is lost when a banknote is completely damaged.

Our solution leverages smart contracts to provide such a mechanism, and to address all of the above issues. It constitutes the first example of the use of smart contracts in conjunction with quantum tools. We suspect that this is not an isolated example, but that smart contracts could find many other useful applications in quantum cryptography.

\paragraph{Blockchains and smart contracts as ideal functionalities}
The essential properties of blockchains and smart contracts can be abstracted by modeling them as ideal functionalities in the Universal Composability (UC) framework of Canetti \cite{canetti2001universally}. Such a framework provides both a formal model for multiparty computation, and formal notions of security with strong composability properties. The approach of studying blockchains and smart contracts within such a framework was first proposed by Bentov al. in \cite{bentov2014use}, and explored further in \cite{bentov2017instantaneous}. The main reason why such an approach is desirable is that it abstracts the features of blockchains and smart contracts into building blocks that can be utilized to design more complex protocols in a modular way. The limitation of the works \cite{bentov2014use,bentov2017instantaneous} is that they modify the original model of computation of UC in order to incorporate coins, but they do not prove that a composition theorem holds in this variant. A more natural approach was proposed by Kiayias et al. in \cite{kiayias2016fair}. There, the authors define a global ideal functionality which abstracts the properties of a transaction-ledger and fits into the well-studied framework of Generalized Universal Composability (GUC) of Canetti et al. \cite{canetti2007global}, and thus inherits its composability properties. Other works \cite{garay2015bitcoin, pass2017analysis} have operated within a similar framework, and have focused on analyzing security of the Bitcoin protocol. In particular, to the best of our knowledge, Badertscher et al. \cite{badertscher2017bitcoin} are the first to propose an ideal functionality for a Bitcoin-like transaction ledger within the GUC framework, together with a secure realization on (an abstraction of) the real-world Bitcoin network.

\paragraph{Our contributions} Our contribution is two-fold: first, we introduce a new global ideal functionality in the GUC framework. Like the previously mentioned works, this functionality abstracts the features of a transaction ledger, but, additionally, is capable of handling a particular type of smart contracts called \textit{stateful} smart contracts. Our approach is inspired by the previous work of Bentov et al. \cite{bentov2017instantaneous} (which also considers stateful smart contracts) but differs primarily in that our ideal functionality fits into the framework of GUC, for which composition theorems are known. Our ideal functionality is arguably simpler to analyze and more coarse-grained than the functionalities in \cite{kiayias2016fair} and \cite{badertscher2017bitcoin}, but also stronger in that it handles stateful smart contracts. Since Bitcoin's scripting language does not allow for stateful smart contracts, our functionality is not realizable on the Bitcoin blockchain. We believe, however, that our functionality, or close variations of it, should have a GUC-secure realization on the Ethereum blockchain, but we do not explicitly provide such a realization (as we expect such a realization to require a fine-grained analysis of the workings of Ethereum). Rather, we assume access to such a functionality, and we focus in this work on its application in a quantum setting. We emphasize, however, that the composition theorems of GUC guarantee that one can replace our ideal functionality with any GUC-secure realization of it on any blockchain, while maintaining security of the protocols that employ it. Our analysis is thus agnostic to the particular real-world implementation.

Second, we provide, to the best of our knowledge, the first example of the use of smart contracts in a quantum setting. We design a hybrid classical-quantum payment system that uses quantum states as banknotes and a classical blockchain to settle disputes and to keep track of the valid serial numbers. This has several desirable features:
\begin{itemize}
\item It is decentralized, requiring no trust in any single entity.
\item Payments are as quick as quantum communication, regardless of the total number of users.
\item When a quantum banknote is damaged or lost, the rightful owner can recover the lost value.
\end{itemize}

As mentioned earlier, the main quantum ingredient that we employ is quantum lightning \cite{zhandry2017quantum}, a strengthening of public-key quantum money with the property that no generation procedure (not even the honest one) can create two banknotes with the same serial number except with negligible probability. As suggested by Zhandry, this property is desirable if one wants to design a decentralized payment system, as it prevents anyone from copying banknotes (even a hypothetical bank). However, since the generation procedure is publicly known, there needs to be a mechanism that regulates the generation of new valid banknotes (to prevent parties from continuously generating banknotes).

We show how stateful smart contracts on a classical blockchain can be used to provide such a mechanism. For instance, they allow to easily keep track of a publicly trusted list of valid serial numbers. We elaborate on this. In our payment system, the native coin of a classical blockchain is used as a baseline classical currency. Any party can spend coins on the classical blockchain to add a serial number of their choice to the list of valid serial numbers. More precisely, they can deposit any amount (of their choice) $d$ of coins into an appropriately specified smart contract and set the initial value of a \textit{serial number} state variable to whatever they wish (presumably the serial number of a quantum banknote that they have just generated locally). They have thus effectively added to the blockchain a serial number associated to a quantum banknote that, in virtue of this, we can think of as having ``acquired'' value $d$. Payments are made by transferring quantum banknotes: party $A$, the payer, transfers his quantum banknote with serial number $s$ to party $B$, the payee, and references a smart contract whose serial number state variable is $s$. Party $B$ then locally verifies that the received banknote is valid. The appeal of such a payment protocol is that transactions only involve the payer and the payee: no consensus mechanism is required to validate the payment, and the payment need not be recorded on the blockchain. The latter is only invoked when a banknote is first generated, and in case of a dispute. Thus, throughput of transactions is no longer a concern. Likewise, long waiting times between when a payment is initiated and when it is confirmed are also no longer a concern since payments are verified immediately by the payee.
The full payment system, which is described formally in section \ref{sec: main}, includes a mechanism that allows any party in possession of a quantum banknote to recover the coins deposited in the corresponding smart contract. It also includes a mechanism that allows any honest party who has lost or damaged a valid quantum banknote to change the \textit{serial number} state variable of the corresponding smart contract to a fresh value of their choice.
\vspace{2mm}
We should note that there is currently only one known construction of quantum lightning, by Zhandry \cite{zhandry2017quantum}, who proved its security under a computational assumption related to the multi-collision resistance of certain degree-2 hash functions (which is arguably not very well studied). Zhandry also shows that secure quantum lightning can be constructed using any non-collapsing hash function, but no example of a provably non-collapsing hash function is currently known (for more details we refer to \cite{zhandry2017quantum}). In our work, we assume existence of a secure quantum lightning scheme. Moreover, we note that in order to be realized and to be practical, our payment system requires, among other quantum technologies, the ability to store quantum states for extended periods of time, often referred to as ``quantum memory''. It also requires the ability of each party to send quantum states to any party that it wishes to make payments to. Hence, it is unlikely that our scheme will be realizable in the near future. The second requirement could be achieved, for example, in a network in which any two parties have the ability to request joint EPR pairs (i.e. maximally entangled pairs of qubits), which is one of the primary components of a ``quantum internet'' \cite{wehner2018quantum}. One can view our scheme as a possible use case of a quantum internet.
One drawback (or advantage?) of our payment system is that, much like the classical banknotes that we currently use, it does not provide the payer with a receipt or a proof that the payment has happened. Payments leave no trace on the blockchain. In virtue of this, a malicious party who receives a quantum banknote can always claim that they did not receive it, when in fact they did. This shortcoming can be mitigated by having payers split their payment into quantum banknotes of smaller value, and requiring the payee to provide receipts for previous parts of the payment, before proceeding with the next. (Or alternatively, payers can just be careful with who they send their quantum banknotes to.)

Finally, we point out that, while it is the first one to use quantum resources, our proposal is not the first that tries to increase the transaction throughput by designing a scheme in which most of the transactions happen off-chain. Most notably, there are two completely classical proposals, called Lightning Network \cite{lightningNetwork2016} (where the terminology is unrelated to quantum lightning!) and Raiden Network, whose aim is to increase throughput of transactions by having most of the them happen off-chain. While promising proposals, their feasibility is much more complex to analyze. As of April 2019, the Lightning Network is being tested at a small scale, and it is still unclear if such a proposal will be practical at full scale. One example of a drawback of the Lightning Network in its peer-to-peer form is that it requires users to deposit into several smart contracts, and for an extended period of time, more coins than they actually expect to spend. 

\paragraph{Outline} Section \ref{sec: prelim} covers preliminaries: \ref{sec: 2-1} covers basic notation; \ref{sec: lightning} introduces quantum lightning; \ref{sec: uc} gives a concise overview of the Universal Composability framework of Canetti \cite{canetti2001universally}. Section \ref{sec: blockchains} gives first an informal description of blockchains and smart contracts, followed by a formal definition of our global ideal functionality for a transaction ledger that handles stateful smart contracts. Section \ref{sec: main} describes our payment system. In section \ref{sec: security}, we describe an adversarial model, and then prove security guarantees with respect to it. 

\section{Preliminaries} 
\label{sec: prelim}

\subsection{Basic notation}
\label{sec: 2-1}
For a function $f: \mathbb{N} \rightarrow \mathbb{R}$, we say that $f$ is \textit{negligible}, and we write $f(n) = negl(n)$, if for any positive polynomial $p(n)$ and all sufficiently large $n$'s, $f(n) < \frac{1}{p(n)}$. A \textit{binary} random variable is a random variable over $\{0,1\}$. We say that two ensembles of binary random variables $\{X_n\}$ and $\{Y_n\}$ are \textit{indistinguishable} if, 
$$\left| \, \Pr[X_n = 1] - \Pr[Y_n = 1]\,\right| = negl(n).$$ We use the term PPT as an abbreviation for probabilistic polynomial time.


\subsection{Quantum money and quantum lightning}
\label{sec: lightning}
Quantum money is a theoretical form of payment first proposed by Wiesner \cite{wiesner1983conjugate}, which replaces physical banknotes with quantum states. In essence, a quantum money scheme consists of a generation procedure, which mints banknotes, and a verification procedure, which verifies the validity of minted banknotes. A banknote consists of a quantum state together with an associated serial number. The appeal of quantum money comes primarily from a fundamental theorem in quantum theory, the No-Cloning theorem, which informally states that there does not exist a quantum operation that can clone arbitrary states. A second appealing property of quantum money, which is not celebrated nearly as much as the first, is that quantum money can be transferred almost instantaneously (by quantum teleportation for example). The first proposals for quantum money schemes required a central bank to carry out both the generation and verification procedures. The idea of public key quantum money was later formalized by Aaronson \cite{aaronson2009quantum}. In public-key quantum money, the verification procedure is public, meaning that anyone with access to a quantum banknote can verify its validity.

In this section, we focus on quantum lightning, a primitive recently proposed by Zhandry \cite{zhandry2017quantum},  and we enhance this to a decentralized quantum payment system. Informally, a quantum lightning scheme is a strengthening of public-key quantum money. It consists of a public generation procedure and a public verification procedure which satisfy the following two properties:
\begin{itemize}
\item Any quantum banknote generated by the honest generation procedure is accepted with probability negligibly close to $1$ by the verification procedure.
\item No adversarial generation procedure (not even the honest one) can generate two banknotes with the same serial number which both pass the verification procedure with non-negligible probability.
\end{itemize}

As mentioned earlier, there is only one known construction of quantum lightning, by Zhandry \cite{zhandry2017quantum}, who gives a construction which is secure under a computational assumption related to the multi-collision resistance of some degree-2 hash function. Zhandry also proves that any non-collapsing hash function can be used to construct quantum lightning. However, to the best of our knowledge, there are no known hash functions that are proven to be non-collapsing. In this section, we define quantum lightning formally, but we do not discuss any possible construction. Rather, in section \ref{sec: main}, we will use quantum lightning as an off-the-shelf primitive.
\begin{definition}[Quantum lightning \cite{zhandry2017quantum}]
A quantum lightning scheme consists of a PPT setup procedure $\textsf{QL.Setup}(1^\lambda)$ (where $\lambda$ is a security parameter) which samples a pair of polynomial-time quantum algorithms $(\textnormal{Gen}$, $\textnormal{Ver})$. $\textnormal{Gen}$ samples states $\ket{\psi} \in \mathcal{H}_{\$}$, which we refer to as ``bolts''. $\textnormal{Ver}$ takes as input a state in $\mathcal{H}_{\$}$ and outputs either $s \in \{0,1\}^{\lambda}$ (the ``serial number'') or $\perp$. They satisfy the following:
\begin{itemize}
    \item Let $\ket\psi$ be any bolt generated by $\textnormal{Gen}$, and let $$H_{min}(\ket{\psi}, \textnormal{Ver}) := -\log_2 \max_s \Pr[\textnormal{Ver}(\ket{\psi}) = s].$$ Then, $$\mathbb{E}[H_{min}(\ket{\psi}, \textnormal{Ver})] = \textnormal{negl}(\lambda),$$
    where the expectation is taken over $(\textnormal{Gen}, \textnormal{Ver}) \leftarrow \textsf{QL.Setup}(1^\lambda)$ and $\ket{\psi} \leftarrow \textnormal{Gen}$.
\end{itemize}
\end{definition}
The latter requirement simply asks that for any honestly generated bolt, there is a serial number $s$ which, with overwhelming probability, is the output of the verification procedure applied to the bolt.
For security, we require that no adversarial generation procedure can produce two bolts with the same serial number. Formally, we define security via the following game between a challenger and an adversary $\mathcal{A}$.

\begin{itemize}
    \item The challenger runs $(\textnormal{Gen}, \textnormal{Ver}) \leftarrow \textsf{QL.Setup}(\lambda)$ and sends $(\textnormal{Gen}, \textnormal{Ver})$ to $\mathcal{A}$.
    \item $\mathcal{A}$ produces a state $\ket{\Psi_{12}} \in \mathcal{H}_{\$}^{\otimes 2}$.
    \item The challenger runs $\textnormal{Ver} \otimes \textnormal{Ver}$ on $\ket{\Psi_{12}}$ obtaining serial numbers $s_0, s_1$. The output of the game is $1$ if $s_0 = s_1 \neq \perp$, and $0$ otherwise.
\end{itemize}

We let $\textnormal{Copy}(\lambda, \mathcal{A})$ be the random variable which denotes the output of the game.

\begin{definition}[Security \cite{zhandry2017quantum}]
A quantum lightning scheme is secure if, for all polynomial-time quantum adversaries $\mathcal{A}$, $$Pr[\textnormal{Copy}(\lambda, \mathcal{A}) = 1 ] = \textnormal{negl}(\lambda).$$
\end{definition}

We define an additional property of a quantum lightning scheme, which in essence establishes that one can trade a quantum banknote for some useful classical certificate. Intuitively, this is meant to capture the fact that in the construction of quantum lightning based on non-collapsing hash functions proposed by Zhandry, one can measure a bolt with serial number $y$ in the computational basis to obtain a pre-image of $y$ under the hash function. However, doing so damages the bolt so that it will no longer pass verification. In order to define this additional property, we change the procedure $\textsf{QL.Setup}(1^{\lambda})$ slightly, so that it outputs a triple $(\textnormal{Gen}, \textnormal{Ver}, H)$ where $H: \{0,1\}^{l(\lambda)}\rightarrow \{0,1\}^{\lambda} \cup \{\perp\}$ (for some polynomial $l(\lambda)$). The additional property is defined based on the following two games $\textsf{Game}_1$, $\textsf{Game}_2$ between a challenger and an adversary $\mathcal{A}$. 

\noindent In $\textsf{Game}_1$:
\begin{itemize}
    \item The challenger runs $(\textnormal{Gen}, \textnormal{Ver}, H) \leftarrow \textsf{QL.Setup}(1^\lambda)$ and samples $\ket{\psi}\leftarrow \textnormal{Gen}$. Let $s := \textnormal{Ver}(\ket{\psi})$. The challenger sends $(\textnormal{Gen}, \textnormal{Ver}, H, \ket{\psi})$ to $\mathcal{A}$.
    \item $\mathcal{A}$ returns $x \in \{0,1\}^{\lambda}$ and $\ket{\psi'}$.
     \item $\mathcal{A}$ wins if $H(x) = s$.
\end{itemize}

\noindent In $\textsf{Game}_2$:
\begin{itemize}
    \item The challenger runs $(\textnormal{Gen}, \textnormal{Ver}, H) \leftarrow \textsf{QL.Setup}(1^\lambda)$ and sends $(\textnormal{Gen}, \textnormal{Ver}, H)$ to $\mathcal{A}$.
    \item $\mathcal{A}$ returns $x \in \{0,1\}^{\lambda}$ and $\ket{\psi}$.
    \item $\mathcal{A}$ wins if $H(x) = \textnormal{Ver}(\ket{\psi})$.
\end{itemize}

Let $\textsf{Game}_1(\mathcal{A}, \lambda)$ (resp. $\textsf{Game}_2(\mathcal{A}, \lambda)$) be the random variable that is $1$ if $\mathcal{A}$ wins $\textsf{Game}_1$ (resp. $\textsf{Game}_2(\mathcal{A}, \lambda)$), and is $0$ otherwise.

\begin{definition}[Trading the bolt for a classical certificate]
\label{def: extra property}
Let $\lambda \in \mathbb{N}$. We say that a quantum lightning scheme has ``bolt-to-certificate capability'' if:

\begin{itemize}
\item[(I)] There exists a polynomial-time quantum algorithm $\mathcal{A}_*$ such that $$\Pr[\textsf{Game}_1(\mathcal{A}_*, \lambda)= 1] = 1-negl(\lambda)$$
\item[(II)] For all polynomial-time quantum algorithms $\mathcal{A}$, $$ \Pr[ \textsf{Game}_{12}(\mathcal{A}, \lambda)= 1 ] = negl(\lambda).$$
\end{itemize}
\end{definition}

Notice that property $(II)$ also implies that for most $H$ and $s$ it is hard for any adversary to find $x$ such that $H(x) = s$ without access to a valid state whose serial number is $s$. In fact, if there was an adversary $\mathcal{A}$ which suceeded at that, this could clearly be used to to construct an adversary $\mathcal{A'}$ that succeeds in $\textsf{Game}_{2}$: Upon receiving $(\textnormal{Gen}, \textnormal{Ver}, H)$ from the challenger, $\mathcal{A}'$ computes generates $\ket{\psi} \leftarrow Gen$ and computes $s:= Ver(\ket{\psi})$; then runs $\mathcal{A}$ on input $s$ to obtain some $x$. $\mathcal{A}'$ returns $x$ and $\ket{\psi}$ to the challenger. We emphasize that in $\textsf{Game}_2$ it is the adversary himself who generates the state $\ket{\psi}$. This is important because when we employ the quantum lightning scheme later on parties are allowed to generate their own quantum banknotes.

\begin{prop}
\label{prop: extra property 1}
Any scheme that uses Zhandry's construction \cite{zhandry2017quantum} instantiated with a non-collapsing hash function satisfies the property of Definition \ref{def: extra property}. 
\end{prop}
\begin{proof}
In Zhandry's construction based on a non-collapsing hash function, $\textsf{QL.Setup}(1^{\lambda})$ outputs $(\textnormal{Gen}, \textnormal{Ver}, H)$, where $n \in \mathbb{N}$ is polynomial in the security parameter $\lambda$, and $H = \tilde{H}^{\times n}$ is the $n$-fold product (or concatenation) of a non-collapsing hash-function $\tilde{H}$, i.e. $H(x_1, ..,x_n) := (\tilde{H}(x_1),..,\tilde{H}(x_n))$. A bolt generated from \textnormal{Gen} has the form $\ket{\Psi} = \bigotimes_{i=1}^n \ket{\psi_{y_i}}$, where $y_i \in \{0,1\}^{\lambda}$ for all $i$, and $\ket{\psi_{y_i}} = \sum_{x: \tilde{H}(x) = y_i} \ket{x}$, and \textnormal{Ver} has the form of $n$-fold product of a verification procedure $\textnormal{Mini-Ver}$ which acts on a single register. 

To show that $(I)$ holds, take $\mathcal{A}_*$ to be the following: $\mathcal{A}_*$ measures each register of the received state in the computational basis to obtain pre-images $x_i$, for $i=1,..,n$. Then, returns $x = (x_1,..,x_n)$ to the challenger (notice that such $\mathcal{A}_*$ clearly destroys the received state, and does not help in winning $\textsf{Game}_{12}$). For property $(II)$, suppose there exists $\mathcal{A}$ such that $\Pr[\textsf{Game}_{12}(\mathcal{A}, \lambda)= 1]$ is non-negligible. We use $\mathcal{A}$ to construct an adversary $\mathcal{A}'$ that breaks collision-resistance of a non-collapsing hash function as follows: $\mathcal{A}'$ runs $(\textnormal{Gen}, \textnormal{Ver}, H) \leftarrow \textsf{QL.Setup}(1^\lambda)$ (where $H = \tilde{H}^{\times n}$ for some non-collapsing hash function $\tilde{H}$), and samples $\ket{\psi}\leftarrow \textnormal{Gen}$. $\mathcal{A}'$ gives $(\textnormal{Gen}, \textnormal{Ver}, H, \ket{\psi})$ as input to $\mathcal{A}$. $\mathcal{A}$ returns $x$ which is parsed as $(x_1,..,x_n)$ and $\ket{\psi'}$. $\mathcal{A}'$ then measures each of the $n$ registers of $\ket{\psi'}$ to get $x' = (x'_1,..,x'_n)$. If $x_i \neq x_i'$, then $\mathcal{A}'$ outputs $(x_i, x_i')$. With non-neglibile probability this pair is a collision for $H$: in fact, since $\mathcal{A}$ wins $\textsf{Game}_{12}$ with non-negligible probability, then $\ket{\psi'}$ must pass Ver with non-negligible probability; and from the analysis of \cite{zhandry2017quantum}, any such state must be such that at least one of the registers is in a superposition which has non-negligible weight on at least two pre-images.
\end{proof}

\begin{prop}
\label{prop: extra property 2}
Zhandry's construction based on the multi-collision resistance of certain degree-2 hash functions (from section 6 of \cite{zhandry2017quantum}) satisfies the property of Definition \ref{def: extra property}. 
\end{prop}
\begin{proof}
The proof is similar to the proof of Proposition \ref{prop: extra property 1}. We include it for completeness in the appendix (\ref{sec: appendix}). 
\end{proof}

\subsection{Universal Composability}
\label{sec: uc}
This subsection is intended as a concise primer about the Universal Composability (UC) model of Canetti \cite{canetti2001universally}. We refer the reader to \cite{canetti2001universally} for a rigorous definition and treatment of the UC model and of composable security, and to the tutorial $\cite{canetti2006tutorial}$ for a more gentle introduction. At the end, we provide a brief overview of the Generalized UC model (GUC) \cite{canetti2007global}. While introducing UC and GUC, we also setup some of the notation that we will employ in the rest of the paper. The reader familiar with these concepts may wish to skip ahead to the next section. 

In the universal composability framework (UC), parties are modelled as Interactive Turing Machines (ITM) who can communicate by writing on each other's externally writable tapes, subject to some global constraints.  Informally, a protocol is specified by a code $\pi$ for an ITM, and consists of various rounds of communication and local computation between instances of  ITMs (the parties), each running the code $\pi$ on their machine, on some private input.

Security in the UC model is defined via the notion of emulation. Informally, we say that a protocol $\pi$ emulates a protocol $\phi$ if whatever can be achieved by an adversary attacking $\pi$ can also be achieved by some other adversary attacking $\phi$. This is formalized by introducing simulators and environments.

Given protocols $\pi$ and $\phi$, we say that $\pi$ emulates (or "is as secure as") $\phi$ in the UC model, if for any polynomial-time adversary $\mathcal{A}$ attacking protocol $\pi$, there exists a polynomial-time simulator $\mathcal{S}$ attacking $\phi$ such that no polynomial-time distinguisher $\mathcal{E}$, referred to as the \textit{environment}, can distinguish between $\pi$ running with $\mathcal{A}$ and $\phi$ running with $\mathcal{S}$. Here, the environment $\mathcal{E}$ is allowed to choose the protocol inputs, read the protocol outputs, including outputs from the adversary or the simulator, and to communicate with the adversary or simulator during the execution of the protocol (without of course being told whether the interaction is with the adversary or with the simulator). In this framework, one can formulate security of a multiparty cryptographic task by first defining an \textit{ideal functionality} $\mathcal{F}$ that behaves exactly as intended, and then providing a ``real-world'' protocol that emulates, or ``securely realizes'', the ideal functionality $\mathcal{F}$. 

We give more formal definitions for the above intuition. To formulate precisely what it means for an environment $\mathcal{E}$ to tell two executions apart, one has to formalize the interaction between $\mathcal{E}$ and the protocols in these executions. Concisely, an execution of a protocol $\pi$ with adversary $\mathcal{A}$ and environment $\mathcal{E}$ consists of a sequence of activations of ITMs. At each activation, the active ITM runs according to its code, its state and the content of its tapes, until it reaches a special \textsf{wait} state. The sequence of activations proceeds as follows: The environment $\mathcal{E}$ gets activated first and chooses inputs for $\mathcal{A}$ and for all parties. Once $\mathcal{A}$ or a party is actived by an incoming message or an input, it runs its code until it produces an outgoing message for another party, an output for $\mathcal{E}$, or it reaches the $\textsf{wait}$ state, in which case $\mathcal{E}$ is activated again. The execution terminates when $\mathcal{E}$ produces its output, which can be taken to be a single bit. Note that each time it is activated, $\mathcal{E}$ is also allowed to invoke a new party, and assign a unique PID (party identifier) to it. 
Allowing the environment to invoke new parties will be particularly important in section \ref{sec: security}, where we discuss security. There, the fact that the environment has this ability implies that our security notion captures realistic scenarios in which the set of parties is not fixed at the start, but is allowed to change. Moreover, each invocation of a protocol $\pi$ is assigned a unique session identifier SID, to distinguish it from other invocations of $\pi$. We denote by $\textrm{EXEC}_{\pi,\mathcal{A}, \mathcal{E}}(\lambda, z)$ the output of environment $\mathcal{E}$ initialized with input $z$, and security parameter $\lambda$ in an execution of $\pi$ with adversary $\mathcal{A}$.  

We are ready to state the following (slightly informal) definition.

\begin{definition}
\label{def: emulation}
A protocol $\pi$ UC-emulates a protocol $\phi$ if, for any PPT adversary $\mathcal{A}$, there exists a PPT simulator $\mathcal{S}$ such that, for any PPT environment $\mathcal{E}$, the families of random variables $\{\textrm{EXEC}_{\pi,\mathcal{A}, \mathcal{E}}(\lambda, z)\}_{\lambda \in \mathbb{N}, z \in \{0,1\}^{poly(\lambda)}}$ and $\{\text{EXEC}_{\phi,\mathcal{S}, \mathcal{E}}(\lambda, z)\}_{\lambda \in \mathbb{N}, z \in \{0,1\}^{poly(\lambda)}}$ are indistinguishable.
\end{definition}

Then, given an ideal functionality $\mathcal{F}$ which captures the intended ideal behaviour of a certain cryptographic task, one can define the ITM code $I_{\mathcal{F}}$, which behaves as follows: the ITM running $I_\mathcal{F}$ simply forwards any inputs received to the ideal functionality $\mathcal{F}$. We then say that a ``real-world'' protocol $\pi$ securely realizes $\mathcal{F}$ if $\pi$ emulates $I_{\mathcal{F}}$ according to Definition \ref{def: emulation}.

\paragraph{A composition theorem} The notion of security we just defined is strong. One of the main advantages of such a security definition is that it supports composition, i.e. security remains when secure protocols are executed concurrently, and arbitrary messages can be sent between executions. We use the notation $\sigma^{\pi}$ for a protocol $\sigma$ that makes up to polynomially many calls to another protocol $\pi$. In a typical scenario, $\sigma^{\mathcal{F}}$ is a protocol that makes use of an ideal functionality $\mathcal{F}$, and $\sigma^{\pi}$ is the protocol that results by implementing $\mathcal{F}$ through the protocol $\pi$ (i.e. replacing calls to $\mathcal{F}$ by calls to $\pi$). It is natural to expect that if $\pi$ securely realizes $\mathcal{F}$, then $\sigma^{\pi}$ securely realizes $\sigma^{\mathcal{F}}$. This is the content of the following theorem.

\begin{theorem}[Universal Composition Theorem]
Let $\pi$, $\phi$, $\sigma$ be polynomial-time protocols. Suppose protocol $\pi$ UC-emulates $\phi$. Then $\sigma^{\pi}$ UC-emulates $\sigma^{\phi}$.
\end{theorem}

Replacing $\phi$ by $I_{\mathcal{F}}$ for some ideal functionality $\mathcal{F}$ in the above theorem yields the composable security notion discussed above.


\paragraph{Generalized UC model}
The formalism of the original UC model is not able to handle security requirements in the presence of a ``global trusted setup''. By this, we mean some global information accessible to all parties, which is guaranteed to have certain properties.  Examples of this are a public-key infrastructure or a common reference string. Emulation in the original UC sense is not enough to guarantee composability properties in the presence of a global setup. Indeed, one can construct examples in which a UC-secure protocol for some functionality interacts badly with another UC-secure protocol and affects its security, if both protocols make reference to the same global setup. For more details and concrete examples see \cite{canetti2007global}.

The generalized UC framework (GUC) of Canetti et al. \cite{canetti2007global} allows for a ``global setup''. The latter is modelled as an ideal functionality which is allowed to interact not only with the parties running the protocol, but also with the environment. GUC formulates a stronger security notion, which is sufficient to guarantee a composition theorem, i.e. ideal functionalities with access to a shared global functionality $\mathcal{G}$ can be replaced by protocols that securely realize them in the presence of $\mathcal{G}$. Further, one can also replace global ideal functionalities with appropriate protocols realizing them. This kind of replacement does not immediately follow from the previous composition theorem and requires a more careful analysis, as is done in \cite{canetti2016pki}, where sufficient conditions for this replacement are established.

\paragraph{Universal Composability in the quantum setting}
In our setting, we are interested in honest parties, adversaries and environments that are quantum polynomial-time ITMs. The notion of Universal Composability has been studied in the quantum setting in \cite{benor2004general}, \cite{unruh2004simulatable} and \cite{unruh2010universally}. In particular, in \cite{unruh2010universally}, Unruh extends the model of computation of UC and its composition theorems to the setting in which polynomial-time classical ITMs are replaced by polynomial-time quantum ITMs (and ideal functionalities are still classical). The proofs are essentially the same as in the classical setting. Although the quantum version of the Generalized UC framework has not been explicitly studied in \cite{unruh2010universally}, one can check that the proofs of the composition theorems for GUC from \cite{canetti2007global} and \cite{canetti2016pki} also go through virtually unchanged in the quantum setting.

\section{Blockchains and smart contracts}
\label{sec: blockchains}
In this section, we start by describing blockchains and smart contracts informally. We follow this by a more formal description. As mentioned in the introduction, the essential features of blockchains and smart contracts can be abstracted by modeling them as ideal functionalities in the Universal Composability framework of Canetti \cite{canetti2001universally}. In this section, we introduce a global ideal functionality that abstracts the properties of a transaction ledger capable of handling stateful smart contracts. We call this $\mathcal{F}_{Ledg}$, which we describe in Fig. \ref{fig: ledger functionality}.
\vspace{1.5mm}

Informally, a blockchain is a public ledger consisting of a sequence of blocks. Each block typically contains information about a set of transactions, and a new block is appended regularly via a consensus mechanism that involves the nodes of a network. A blockchain is equipped with a native currency which is employed in transactions, and whose basic unit is referred to as a \textit{coin}.

Each user in the network is associated with a public key (this can be thought of as the user's address). A typical transaction is a message which transfers coins from a public key to another. It is considered valid if it is digitally signed using the secret key corresponding to the sending address. 

More precisely, in Bitcoin, parties do not keep track of users's accounts, but rather they just maintain a local copy of a set known as ``unspent transaction outputs set'' (UTXO set). An unspent output is a transaction that has not yet been ``claimed'', i.e. the coins of these transactions have not yet been spent by the receiver. Each unspent output in the UTXO set includes a circuit (also known as a ``script'') such that any user that can provide an input which is accepted by the circuit (i.e. a witness) can make a transaction that spends these coins, thus creating a new unspent output. Hence, if only one user knows the witness to the circuit, he is effectively the owner of these coins. For a standard payment transaction, the witness is a signature and the circuit verifies the signature. However, more complex circuits are also allowed, and these give rise to more complex transactions than simple payments: smart contracts. A smart contract can be thought of as a transaction which deposits coins to an address. The coins are released upon fulfillment of certain pre-established conditions. 

In \cite{bentov2014use} and \cite{bentov2017instantaneous}, smart contracts are defined as ideal functionalities in a variant of the Universal Composability (UC) model \cite{canetti2001universally}. The ideal functionality that abstracts the simplest smart contracts was formalized in \cite{bentov2014use}, and called ``Claim or Refund''. Informally, this functionality specifies that a sender $P$ locks his coins and chooses a circuit $\phi$, such that a receiver $Q$ can gain possession of these coins by providing a witness $w$ such that $\phi(w) = 1$ before an established time, and otherwise the sender can reclaim his coins. The ``Claim or Refund'' ideal functionality can be realized in Bitcoin as long as the circuit $\phi$ can be described in Bitcoin's scripting language. On the other hand, Ethereum's scripting language is Turing-complete, and so any circuit $\phi$ can be described.

``Claim or refund'' ideal functionalities can be further generalized to ``stateful contracts''. In Ethereum, each unspent output also maintains a \textit{state}. In other words, each unspent output comprises not only a circuit $\phi$, but also state variables. Parties can claim partial amounts of coins by providing witnesses that satisfy the circuit $\phi$ in accordance with the current state variables. In addition, $\phi$ also specifies an update rule for the state variables, which are updated accordingly. We refer to these type of transactions as stateful contracts, as opposed to the ``stateless'' contract of ``Claim or Refund''. 
Stateful contracts can be realized in Ethereum, but not in Bitcoin. From now onwards, we will only work with stateful contracts. We will use the terms ``smart contracts'' and ``stateful contracts'' interchangeably. 

We emphasize that our modeling is inspired by \cite{bentov2014use} and \cite{bentov2017instantaneous}, but differs in the way that coins are formalized. One difference from the model of Bentov et al. is that there, in order to handle coins, the authors augment the original UC model by endowing each party with a \textit{wallet} and a \textit{safe}, and by considering \textit{coins} as atomic entities which can be exchanged between parties. To the best of our knowledge, this variant of the UC framework is not subsumed by any of the previously studied variants, and thus it is not known whether a composition theorem holds for it.

On the other hand, we feel that a more natural approach is to work in the Generalized UC model \cite{canetti2007global}, and to define a global ideal functionality $\mathcal{F}_{Ledg}$ which abstracts the essential features of a \textit{transaction ledger} capable of handling stateful smart contracts. This approach was first proposed by Kiayias et al. \cite{kiayias2016fair}. The appeal of modeling a transaction ledger as a global ideal functionality is that composition theorems are known in the Generalized UC framework. In virtue of this, any secure protocol for some task that makes calls to $\mathcal{F}_{Ledg}$ can be arbitrarily composed while still maintaining security. This means that one need not worry about composing different concurrent protocols which reference the same transaction ledger. One would hope that it is also the case that a secure protocol that makes calls to $\mathcal{F}_{Ledg}$ remains secure when the latter are replaced by calls to secure real-world realizations of it (on Ethereum for example). This requires a more careful analysis, and Canetti et al. provide in \cite{canetti2016pki} sufficient conditions for this replacement to be possible. We do not prove that a secure real-world realization of $\mathcal{F}_{Ledg}$ on an existing blockchain exists, but we believe that $\mathcal{F}_{Ledg}$, or a close variant of it, should be securely realizable on the Ethereum blockchain. In any case, we work abstractly by designing our payment system, and proving it secure, assuming access to such an ideal functionality. The appeal of such an approach is that the security of the higher-level protocols is independent of the details of the particular real-world implementation of $\mathcal{F}_{Ledg}$. 

Next, we describe our global ideal functionality $\mathcal{F}_{Ledg}$. In doing so, we establish the notation that we will utilize in the rest of the paper. 

\vspace{1.5mm}

\paragraph{Global ledger ideal functionality} 


We present in Fig. \ref{fig: ledger functionality} our global ledger ideal functionality $\mathcal{F}_{Ledg}$. In a nutshell, this keeps track of every registered party's coins, and allows any party to transfer coins in their name to any other party. It also allows any party to retrieve information about the number of coins of any other party, as well as about any previous transaction. The initial amount of coins of a newly registered party is determined by its PID (recall that in a UC-execution the PID of each invoked party is specified by the environment; see subsection \ref{sec: uc} for more details). Moreover, $\mathcal{F}_{Ledg}$ handles (stateful) smart contracts: it accepts deposits from the parties involved in a contract and then pays rewards appropriately. Recall that in stateful smart contracts a party or a set of parties deposit an amount of coins to the contract. The contract is specified by a circuit $\phi$, together with an initial value for a state variable $\textsf{st}$. A state transition is triggered by any party $P$ with PID $pid$ sending a witness $w$ which is accepted by $\phi$ in accordance with the current state and the current time $t$. More precisely, the contract runs $\phi(pid,w,t,\textsf{st})$, which outputs either ``$\perp$'' or a new state (stored in the variable \textsf{st}) and a number of coins $d \in \mathbb{N}$ that is released to $P$. Each contract then repeatedly accepts state transitions until it has distributed all the coins that were deposited into it at the start. Notice that $\phi$ can accept different witnesses at different times (the acceptance of a witness can depend on the current time $t$ and the current value of the state variable $\textsf{st}$). Information about existing smart contracts can also be retrieved by any party.

The stateful-contract portion of $\mathcal{F}_{Ledg}$ resembles closely the functionality $\mathcal{F}_{StCon}$ from \cite{bentov2017instantaneous}. Our approach differs from that of $\cite{bentov2017instantaneous}$ in that the we make the smart contract functionality part of the global ideal functionality $\mathcal{F}_{Ledg}$ which also keeps track of party's coins and transactions. In \cite{bentov2017instantaneous} instead, coins are incorporated in the computation model by augmenting the ITMs with wallets and safes (this changes the model of computation in a way that is not captured by any of the the previously studied variants of the UC framework).

We implicitly assume access to an ideal functionality for message authentication $\mathcal{F}_{auth}$ which all parties employ when sending their messages, and also to a global ideal functionality for a clock that keeps track of time. We assume implicitly that $\mathcal{F}_{Ledg}$ makes calls to the clock and keeps track of time. Alternatively, we could just have $F_{Ledg}$ maintain a local variable that counts the number of transactions performed, and a local time variable $t$, which is increased by $1$ every time the number of transactions reaches a certain number, after which the transaction counter is reset (this mimics the process of addition of blocks in a blockchain, and time simply counts the number of blocks). From now onwards, we do not formally reference calls to $\mathcal{F}_{auth}$ or to the clock to avoid overloading notation. We are now ready to define $\mathcal{F}_{Ledg}$.

\begin{figure}[H]
\rule[1ex]{16.5cm}{0.5pt}
{\centering \textbf{Global ledger ideal functionality} \par}

Initialize the sets $\textsf{parties} = \{\}$, $\textsf{contracts} = \{\}$, $\textsf{AllTransactions} = \{\}$. (Throughout, the variable $t$ denotes the current time.)

\paragraph{Register:} Upon receiving a message \textsf{Register} from a party with PID $pid = (id, d)$, send (\textsf{Registered}, $pid$) to the adversary; upon receiving a message \textsf{ok} from the adversary, and if this is the first request from $pid$, add $pid$ to the set $\textsf{parties}$. Set $pid.\textsf{id} \leftarrow id$ and $pid.\textsf{coins} \leftarrow d$.

\paragraph{Retrieve party:} Upon receiving a message (\textsf{RetrieveParty}, $pid$) from some party $P$ (or the adversary), output (\textsf{RetrieveParty}, $pid$, $d$) to $P$ (or to the adversary), where $d = \perp$ if $pid \notin \textsf{parties}$, and $d = pid.\textsf{coins}$ otherwise. (We slightly abuse notation here in that, when taken as part of a message, $pid$ is treated as a string, but, when called by the functionality, $pid$ is a variable with attributes $pid.\textsf{id}$ and $pid.\textsf{coins}$).

\paragraph{Pay:} Upon receiving a message ($\textsf{pay}$, $pid'$, $d$) from some party $P$ with PID $pid$, and $pid \in \textsf{parties}$, do the following:
\begin{itemize}
    \item If $pid' \in \textsf{parties}$ and $pid.\textsf{coins} >d$, update $pid'.\textsf{coins} \leftarrow pid'.\textsf{coins} + d$ and $pid.\textsf{coins}  \leftarrow pid.\textsf{coins} - d$. Set $trId = |\textsf{AllTransactions}| + 1$. Add a variable named $trId$ to $\textsf{AllTransactions}$, with attribute $trId.\textsf{transaction} = (pid, pid', d, t)$. Send a message (\textsf{Executed}, $trId$) to $P$.
    \item Else, return $\perp$ to $P$.
\end{itemize}

\paragraph{Retrieve transaction:} Upon receiving a message (\textsf{RetrieveTransaction}, $trId$) from some party $P$ (or the adversary), output (\textsf{RetrieveTransaction}, $trId$, $s$), where $s = \perp$ if $trId \notin \textsf{allTransactions}$, and $s = trId.\textsf{transaction}$ otherwise.

\paragraph{Smart contract} Upon receiving a message (\textsf{InitiateSmartContract}, \textsf{Params}=$(I, D, \phi, \textsf{st}_0)$), where $I$ is a set of PID's, $D$ is a set $\{(pid, d_{pid}): pid \in I\}$ of ``initial deposits'', with $d_{pid}$ being the amount required initially from the party with PID $pid$, $\phi$ is a circuit, and $\textsf{st}_0$ is the initial value of a state variable $\textsf{st}$, check that $I \subseteq \textsf{parties}$. If not, ignore the message; if yes, set $ssid = |\textsf{contracts}| + 1$. Add a variable named $ssid$ to $\textsf{contracts}$ with attributes $ss id.\textsf{Params} = (I, D, \phi, \textsf{st}_0)$, $ssid.\textsf{state} = \textsf{st}$ and $ssid.\textsf{coins} \leftarrow 0$. Send a message (\textsf{RecordedContract}, $ssid$) to $P$. Then, do the following:
\begin{itemize}
    \item \textbf{Initialization phase:} Wait to get message $(\textsf{InitializeWithCoins}, \textit{ssid}, \textsf{Params}=(I, D, \phi, \textsf{st}_0))$ from party with PID $pid$ for all $pid \in I$. When all messages are received, and if, for all $pid \in I$, $pid.\textsf{coins} \geq d_{pid}$, then, for all $pid \in I$, update: $pid.\textsf{coins} \leftarrow pid.\textsf{coins} - d_{pid}$ and $ssid.\textsf{coins} \leftarrow ssid.\textsf{coins} + d_{pid}$. Set $\textsf{st} \leftarrow \textsf{st}_0$  (We assume that $ssid.\textsf{state}$ changes dynamically with $\textsf{st}$).
    \item \textbf{Execution phase:} Repeat until termination: Upon receiving a message of the form $(\textsf{Trigger}, \textit{ssid}, w, d)$ at time $t$ from some party with PID $pid \in \textsf{parties}$ (where it can also be $d=0$) such that $\phi(pid, w, t, \textsf{st}, d) \neq \perp$, do the following:
    \begin{itemize}
        \item If $d>0$, update $pid.\textsf{coins} \leftarrow pid.\textsf{coins} - d$ and $ssid.\textsf{coins} \leftarrow ssid.\textsf{coins} + d$
        \item Update $(\textsf{st}, e) \leftarrow \phi(pid, w, t, \textsf{st}, d)$. 
        \item If $e = \textnormal{``all coins''}$, let $q := ssid.\textsf{coins}$. Send the message (\textsf{Reward}, $ssid$, $q$) to the party with PID $pid$ and update $pid.\textsf{coins} \leftarrow pid.\textsf{coins} + q$ and $ssid.\textsf{coins} \leftarrow 0$. If $e > 0$ and $ssid.\textsf{coins} \geq e$, send the message (\textsf{Reward}, $ssid$, $e$) to the party with PID $pid$, and update $pid.\textsf{coins} \leftarrow pid.\textsf{coins} + e$ and $ssid.\textsf{coins} \leftarrow ssid.\textsf{coins} -e$. Else, if $ssid.\textsf{coins} = e' < e$, send the message (\textsf{Reward}, $ssid$, $e'$) to the party with PID $pid$, and update $pid.\textsf{coins} \leftarrow pid.\textsf{coins} + e'$ and $ssid.\textsf{coins} \leftarrow 0$. Then, terminate.
    \end{itemize}
\end{itemize}

\paragraph{Retrieve smart contract:} Upon receiving a message (\textsf{RetrieveContract}, $ssid$) from some party $P$ (or the adversary), output (\textsf{RetrieveContract}, $ssid$, $z$), where $z = \perp$ if $ssid \notin \textsf{contracts}$, and $z = (ssid.\textsf{Params}, ssid.\textsf{state}, ssid.\textsf{coins})$ otherwise.

\rule[2ex]{16.5cm}{0.5pt}\vspace{-.5cm}
\caption{Global ledger ideal functionality $\mathcal{F}_{Ledg}$}
  \label{fig: ledger functionality}
 
\end{figure}

From now on, we will often refer to the number of coins $ssid.\textsf{coins}$ of a contract with session identifier $ssid$ as the coins \textit{deposited} in the contract. When we say that a contract $\textit{releases}$ some coins to a party $P$ with PID $pid$, we mean more precisely that $\mathcal{F}_{Ledg}$ updates its local variables and moves coins from $ssid.\textsf{coins}$ to $pid.\textsf{coins}$.

\section{A payment system based on quantum lightning and a classical blockchain}
\label{sec: main}
In this section, we describe our payment system. We give first an informal description, and in section \ref{sec: 3-2} we give a formal description. 

The building block of our payment system is a quantum lightning scheme, reviewed in detail in section \ref{sec: lightning}. Recall that a quantum lightning scheme consists of a generation procedure which creates quantum banknotes, and a verification procedure that verifies them and assigns serial numbers. The security guarantee is that no generation procedure (not even the honest one) can create two banknotes with the same serial number except with negligible probability. As mentioned earlier, this property is desirable if one wants to design a decentralized payment system, as it prevents anyone from cloning banknotes (even the person who generates them). However, this calls for a mechanism to regulate generation of new valid quantum banknotes.

In this section, we describe formally a proposal that employs smart contracts to provide such a mechanism. As we have described informally in the introduction, the high-level idea is to keep track of the valid serial numbers using smart contracts. Any party is allowed to deposit any amount of coins $d$ (of their choice) into a smart contract with specific parameters (see definition \ref{def: banknote-contract} below), and with an initial value of his choice for a $\textit{serial number}$ state variable. We can think of the quantum banknote with the chosen serial number as having ``acquired'' value $d$. A payment involves only two parties: a payer, who sends a quantum banknote, and a payee who receives it and verifies it locally. As anticipated in the introduction, the full payment system includes the following additional features, which we describe here informally in a little more detail (all of these are described formally in subsection \ref{sec: 3-2}):

\begin{itemize}
\item Removing a serial number from the list of valid serial numbers in order to recover the amount of coins deposited in the corresponding smart contract. This makes the two commodities (quantum banknotes and coins on the blockchain) interchangeable. This is achieved by exploiting the additional property of of some quantum lightning scheme from Definition \ref{def: extra property}. Recall that, informally, this property states that there is some classical certificate that can be recovered by measuring a valid quantum banknote, which no efficient algorithm can recover otherwise. The key is that once the state is measured to recover this certificate, it is damaged in a way that it only passes verification with negligible probability (meaning that it can no longer be spent). We allow users to submit this classical certificate to a smart contract, and if the certificate is consistent with the serial number stored in the contract, then the latter releases all of the coins deposited in the contract to the user. 
\item Allowing a party to replace an existing serial number with a new one of their choice in case they lose a valid quantum state (they are fragile after all!). We allow a user $P$ to file a ``lost banknote claim'' by sending a message and some fixed amount of coins $d_0$ to a smart contract whose serial number is the serial number of the lost banknote. The idea is that if no one challenges this claim, then after a specified time $t_{tr}$ user $P$ can submit a message which changes the value of the serial number state variable to a value of his choice and recovers the previously deposited $d_0$ coins. On the other hand, if a user $P$ maliciously files a claim to some contract with serial number $s$, then any user $Q$ who possesses the valid banknote with serial number $s$ can recover the classical certificate from Definition \ref{def: extra property}, and submit it to the contract. This releases all the coins deposited in the contract to $Q$ (including the $d_0$ deposited by $P$ to make the claim). As you might notice, this requires honest users to monitor existing contracts for ``lost banknote claims''. This, however, is not much of a burden if $t_{tr}$ is made large enough (say a week or a month). The requirement of being online once a week or once a month is easy to meet in practice.
\end{itemize}

\subsection{The payment system and its components}
\label{sec: 3-2}

In this section, we describe in detail all of the components of the payment system. It consists of the following: a protocol to generate valid quantum banknotes; a protocol to make a payment; a protocol to file a claim for a lost banknote; a protocol to prevent malicious attempts at filing claims for lost banknotes; and a protocol to trade a valid quantum banknote in exchange for coins.

Recall from the description of $\mathcal{F}_{Ledg}$ that each smart contract is specified by several parameters: $I$ is a set of PIDs of parties who are expected to make the initial deposit, with $\{d_{pid} : pid \in I\}$ being the required initial deposit amounts; a circuit $\phi$ specifies how the state variables are updated and when coins are released; an initial value $\textsf{st}_0$ for the state variable $\textsf{st}$; a session identifier $ssid$.

Let $\lambda \in \mathbb{N}$. From now onwards, assume that $(\textnormal{Gen}, \textnormal{Ver}, H) \leftarrow \textsf{QL.Setup}(\lambda)$, where the latter is the setup procedure of a quantum lightning scheme with bolt-to-certificate capability (i.e. a quantum lightning scheme that satisfies the additional property of Definition \ref{def: extra property}). Let $\mathcal{A}_*$ be as in Definition \ref{def: extra property}, namely $\textsf{Game}_1(\mathcal{A}_*, \lambda) = 1$ with all but negligible probability. In Definition \ref{def: banknote-contract}, we define an instantiation of smart contracts with a particular choice of parameters, which we refer to as banknote-contracts. Banknote-contracts are the building blocks of the protocols that make up our payment system. We describe a banknote-contract informally before giving a formal definition. 

A banknote-contract is a smart contract initialized by a single party, and it has a state variable of the form $\textsf{st} = (\textsf{serial}, \textsf{ActiveLostClaim})$. The party initializes the banknote-contract by depositing a number of coins $d$ and by setting the initial value of $\textsf{serial}$ to any desired value. The banknote-contract handles the following type of requests:
\begin{itemize}
    \item As long as $\textsf{ActiveLostClaim} = \text{``No active claim''}$ (which signifies that there are no currently active lost-banknote claims), any party $P$ can send the message \textsf{BanknoteLost}, together with a pre-established amount of coins $d_0$ to the contract. This will trigger an update of the state variable $\textsf{ActiveLostClaim}$ to reflect the active lost-banknote claim by party $P$. 
    \item As long as there is an active lost-banknote claim, i.e. $\textsf{ActiveLostClaim} = \text{``Claim by $pid$ at time $t$''}$, any party $Q$ can challenge that claim by submitting a message $(\textsf{ChallengeClaim}, x, s')$ to the contract, where $s'$ is a proposed new serial number. We say that $x$ is a valid classical certificate for the current value $s$ of $\textsf{serial}$ if $H(x) =s$. Such an $x$ can be thought of as a proof that whoever is challenging the claim actually possessed a quantum banknote with serial number $s$, and destroyed it in order to obtain the certificate $x$, and thus that the current active lost-banknote claim is malicious. If $x$ is a valid classical certificate for $s$, then $\textsf{serial}$ is updated to the new value $s'$ submitted by $Q$, who also receives all of the coins deposited in the contract (including the $d_0$ coins deposited by the malicious claim).
    \item If party $P$ has previously submitted a lost-banknote claim, and his claim stays unchallenged for time $t_{tr}$, then party $P$ can send a message $(\textsf{ClaimUnchallenged}, s')$ to the contract, where $s'$ is a proposed new serial number. Then the contract returns to $P$ the $d_0$ coins he initially deposited when making the claim, and updates $\textsf{serial}$ to $s'$.
    \item Any party $P$ can submit to the contract a message $(\textsf{RecoverCoins}, x)$. If $x$ is a valid classical certificate for the current value of $s$ of $\textsf{serial}$, then the contract releases to $P$ all the coins currently deposited in the contract. This allows party $P$ to ``convert'' back his quantum banknote into coins.
\end{itemize}

Next, we will formally define banknote-contracts, and then formally describe all of the protocols that make the payment system.

\begin{figure}[H]
\rule[1ex]{16.5cm}{0.5pt}\\
$\phi_{\$}\left(pid, w, t, (\textsf{serial}, \textsf{ActiveLostClaim}), d\right)$ takes as input strings $pid$ and $w$, where $pid$ is meant to be the PID of some party $P$, and we refer to $w$ as the ``witness'', $t \in \mathbb{N}$ denotes the ``current time'' mantained by $\mathcal{F}_{Ledg}$, $(\textsf{serial}, \textsf{ActiveLostClaim})$ is the current value of the state variable, and $d \in \mathbb{N}$ is the number of coins that are being deposited to the smart contract with the current message. $\phi_{\$}$ has hardcoded parameters: $d_0 \in \mathbb{N}$ the amount of coins needed to file a claim for a lost money state, $t_{tr} \in \mathbb{N}$ the time after which an unchallenged claim can be settled ($d_0$ and $t_{tr}$ are fixed constants agreed upon by all parties, and they are the same for all banknote-contracts). The circuit $\phi_{\$}$ outputs new values for the state variables and an amount of coins as follows:

On input $(pid, w, t, (\textsf{serial} = s, \textsf{ActiveLostClaim}), d)$, $\phi_{\$}$ does the following:
\begin{itemize}
\item If $\textsf{ActiveLostClaim} = \text{``No active claim''}$:
\begin{itemize}
\item If $w = \textsf{BanknoteLost}$ and $d = d_0$, then $\phi$ outputs  $\big((\textsf{serial} = s, \textsf{ActiveLostClaim} = \text{``Claim by $pid$ at time $t$''}), 0\big)$ (to symbolize that at time $t$ party with PID $pid$ has claimed to have lost the money state with serial number $s$, and that zero coins are being released).
\item If $w = (\textsf{RecoverCoins}, x)$, where $x \in \{0,1\}^{\lambda}$ and $H(x) = s$, then $\phi$ outputs $\big( (\textsf{serial} = \perp, \textsf{ActiveLostClaim} = \perp), \textnormal{``all coins''}\big)$
\end{itemize}
\item If $\textsf{ActiveLostClaim} = \text{``Claim by $pid'$ at time $t_0$''}$ for some $pid', t_0$: 
\begin{itemize}
\item If $w = (\textsf{ChallengeClaim}, x, s')$,  where $x \in \{0,1\}^{\lambda}$ and $H(x) = s$, and $s' \in \{0,1\}^{\lambda}$ then $\phi$ outputs  $\big((\textsf{serial} = s', \textsf{ActiveLostClaim} = \text{``No active claim''}), d_0 \big)$.
\item If $w = (\textsf{ClaimUnchallenged}, s')$, $pid=pid'$ and $t-t_0 > t_{tr}$, then $\phi$ outputs  $\big((\textsf{serial} = s', \textsf{ActiveLostClaim} = \text{``No active claim''}), d_0 \big)$.
\end{itemize}
\end{itemize}

\rule[2ex]{16.5cm}{0.5pt}\vspace{-.5cm}
\caption{Circuit $\phi_{\$}$ for banknote-contracts}
  \label{fig: circuit phi}
  
\end{figure}

\begin{definition}(Banknote-contract)
\label{def: banknote-contract}
A banknote-contract, is a smart contract on $\mathcal{F}_{Ledg}$ specified by parameters of the following form: $I = \{pid\}$ for some $pid \in [n]$, $D = \{(pid, d_{pid})\}$ for some $d_{pid} \in \mathbb{N}$, $\textsf{st}_0 = (s, \text{``No active claim''})$ for some $s \in \{0,1\}^{\lambda}$, and circuit $\phi = \phi_{\$}$, where $\phi_{\$}$ is defined as in Fig. \ref{fig: circuit phi}.
\end{definition}

For convenience, we denote by $\textsf{serial}$ and $\textsf{ActiveLostClaim}$ respectively the first and second entry of the state variable of a banknote-contract.


\paragraph{Generating valid quantum banknotes} We describe the formal procedure for generating a valid quantum banknote.



\begin{figure}[H]
\rule[1ex]{16.5cm}{0.5pt}\\
Protocol carried out by some party $P$ with PID $pid$.\\

Input of $P$: An integer $d$ such that $pid.\textsf{coins} > d$ in $\mathcal{F}_{Ledg}$ ($d$ is the ``value'' of the prospective banknote).
\begin{itemize}
\item Run $\ket{\psi} \leftarrow \textnormal{Gen}$. Let $s := \textnormal{Ver}(\ket{\psi})$. 
\item Send $\big(\textsf{InitiateSmartContract},  \textsf{Params} \big)$ to $\mathcal{F}_{Ledg}$, where $\textsf{Params} = (\{pid\}, \{(pid, d)\}, \phi, (s, \text{``No active claim''}))$. Upon receipt of a message of the form (\textsf{RecordedContract}, $ssid$), send the message (\textsf{InitializeWithCoins}, $ssid$, \textsf{Params}) to $\mathcal{F}_{Ledg}$.
\end{itemize}

\rule[2ex]{16.5cm}{0.5pt}\vspace{-.5cm}
\caption{Generating a valid banknote}
  \label{fig: protocol gen banknote}
  
\end{figure}

\paragraph{Making a payment} We describe formally the protocol for making a payment in Fig. \ref{fig: protocol making a payment}. Informally, the protocol is between a party $P$, the payer, and a party $Q$, the payee. In order to pay party $Q$ with a bolt whose serial number is $s$, party $P$ sends the valid bolt to party $Q$, the payee, together with the $ssid$ of a smart contract with $\textsf{serial} =s$. Party $Q$ verifies that $ssid$ corresponds to a banknote-contract with $\textsf{serial} = s$, and verifies that the banknote passes verification and has serial number $s$. 

\begin{figure}[H]
\rule[1ex]{16.5cm}{0.5pt}
The protocol is between some party $P$ with PID $pid$(the payer) and a party $Q$ with PID $pid'$ (the payee):\\

Input of $P$: $\ket{\Psi}$, a valid bolt with serial number $s$. $ssid$ the session identifier of a smart contract on $\mathcal{F}_{Ledg}$ such that $ssid.\textsf{state} = (s,\text{``No active claim''})$, and  $ssid.\textsf{coins} = d$. 
\begin{itemize}
\item $P$ sends state $\ket{\Psi}$ to $Q$. $P$ also sends a message ($ssid$, $s$, $d$) to $Q$.
\item $Q$ sends a message (\textsf{RetrieveContract}, $ssid$) to $\mathcal{F}_{Ledg}$. Upon receiving a message (\textsf{RetrieveContract}, $ssid$, $z$) from $\mathcal{F}_{Ledg}$ (where $z = (ssid.\textsf{Params}, ssid.\textsf{state}, ssid.\textsf{coins})$ if $P$ is honest), $Q$ does the following: 
\begin{itemize}
\item If $z = (\textsf{Params}, (s, \text{``No active claim''}), d)$, then $Q$ checks that the parameters $\textsf{Params}$ are of the form of a banknote-contract (from Definition \ref{def: banknote-contract}). If so, runs $\textnormal{Ver}(\ket{\Psi})$ and checks that the outcome is $s$. If so, sends a message $\textsf{accept}$ to $P$.
\item Else, $Q$ aborts.
\end{itemize} 
\end{itemize}

\rule[2ex]{16.5cm}{0.5pt}\vspace{-.5cm}
\caption{Protocol for making and verifying a payment}
  \label{fig: protocol making a payment}
  
\end{figure}

\paragraph{Recovering lost banknotes} 
\label{sec: recovering}

As much as we can hope for experimental progress in the development of quantum memories, for the foreseeable future we can expect quantum memories to only be able to store states for a time on the order of days. It is thus important that any payment system involving quantum money is equipped with a procedure for users to recover the value associated to quantum states that get damaged and become unusable. Either users should be able to ``convert'' quantum money states back to coins on the blockchain, or they should be able, upon losing a quantum banknote, to change the serial number state variable of the associated smart contract to a new serial number (presumably of freshly generated quantum banknote). Here, we describe a protocol for the latter. Later, we describe a protocol for the former.

Informally, a party $P$ who has lost a quantum banknote with serial number $s$ associated to a smart contract with session identifier $ssid$, makes a ``lost banknote claim'' at time $t$ by depositing a number of coins $d_0$ to that banknote-contract. Recall the definition of banknote-contracts from Definition \ref{def: banknote-contract}, and in particular of the circuit $\phi_{\$}$:
\begin{itemize}
\item If party $P$ is honest, then after a time $t_{tr}$ has elapsed, he will be able to update the state variable $\textsf{serial}$ of the banknote-contract from $s$ to $s'$ (where $s'$ is presumably the serial number of a new valid bolt that party $P$ has just generated).
\item If party $P$ is dishonest, and he is claiming to have lost a banknote that someone else possesses, then the legitimate owner can apply $\mathcal{A}_*$ as in Definition \ref{def: extra property} to the legitimate bolt $\ket{\Psi}$ and recover $x$ such that $H(x) = s$. He can then send $x$ to the contract and a new serial number $s'$ (presumably of a freshly generate bolt) and obtain $d_0$ coins from the contract (the $d_0$ coins deposited by $P$ in his malicious claim).
\end{itemize}
We describe the protocol formally in Fig. \ref{fig: protocol making a claim}.

One might wonder whether, in practice, an adversary can instruct a corrupt party to make a ``lost banknote claim'', and then intercept an honest party's classical certificate $x$ before this is posted to the blockchain, and have a second corrupt party post it instead. This attack would allow the adversary to ``steal'' the honest party's value. In our analysis, we do not worry about this, as we assume access to the ideal functionality $\mathcal{F}_{Ledg}$, which, by definition, deals with incoming messages in the order that they are received. We also assume in our adversarial model, specified more precisely in Section \ref{sec: security}, that the adversary does not have any control over the delivery of messages (and their timing). If one assumes a more powerful adversary (with some control over the timing of delivery of messages), then this kind of issue can still be mitigated, for example, in the following manner: the honest party does not directly post the classical certificate $x$, but she instead first posts a commitment to $x$, and she reveals it only at a latter stage. We thank Or Sattath for pointing out this issue.

\begin{figure}[H]
\rule[1ex]{16.5cm}{0.5pt}
Protocol carried out by party $P$ with PID $pid$ for changing the serial number of a smart contract.\\

$P$'s input: $s$ the serial number of a (lost) quantum banknote. $ssid$ the session identifier of a banknote-contract such that $ssid.\textsf{state} = (s, \text{``No active claim''})$.

\begin{itemize}
\item $P$ sends $(\textsf{Trigger}, ssid, \textsf{BanknoteLost}, d_0$) to $\mathcal{F}_{Ledg}$. This updates $ssid.\textsf{state}$ to $(s, \text{``Active claim by $pid$ at time $t$''})$
(where $t$ is the current time mantained by $\mathcal{F}_{Ledg}$), and deposits $d_0$ coins into the contract.
\item After time $t_{tr}$, $P$ sends $(\textsf{Trigger}, ssid, (\textsf{ClaimUnchallenged}, s'), 0)$ to $\mathcal{F}_{Ledg}$. If $P$ was honest then $ssid.\textsf{state}$ is updated to $(s', \text{``No active claim''})$, and $d_0$ coins are released to $P$.
\end{itemize}

\rule[2ex]{16.5cm}{0.5pt}\vspace{-.5cm}
\caption{Protocol for changing the serial number of a smart contract}
  \label{fig: protocol making a claim}
  
\end{figure}

Next, we give a protocol carried out by all parties to prevent malicious attempts at changing the state variable $\textsf{serial}$ of a smart contract. Informally, this involves checking the blockchain regularly for malicious attempts at filing lost-banknote claims. 

Recall that $t_{tr}$ was defined in Definition \ref{def: banknote-contract}. Recall also the definitions of $\mathcal{A}_*$ and $H$ from Definition \ref{def: extra property}.

\begin{figure}[H]
\rule[1ex]{16.5cm}{0.5pt}

Protocol carried out by a party $P$ to prevent malicious attempts at changing the state variable $\textsf{serial}$ of a smart contract.\\

Input of $P$: A triple $(\ket{\Psi}, s, ssid)$, where $\ket{\Psi}$ is a quantum banknote with serial number $s$, and $ssid$ is the session identifier of a banknote-contract such that $ssid.\textsf{state} = (s, \text{``No active claim''})$ \\

At regular intervals of time $t_r - 1$, do the following: 
\begin{itemize}
\item Send a message (\textsf{RetrieveContract}, $ssid$) to $\mathcal{F}_{Ledg}$. Upon receiving a message (\textsf{RetrieveContract}, $ssid$, $z$) from $\mathcal{F}_{Ledg}$, if $z =(\textsf{Params}, (s, \text{``Claim by $pid'$ at time $t$ ''}), d)$ for some $pid', t, d$ and for some banknote-contract parameters $\textsf{Params}$:
\begin{itemize}
    \item Apply $\mathcal{A}_*$ to $\ket{\Psi}$ to obtain $x \in \{0,1\}^{\lambda}$ (such that $H(x) = s$).
    \item Sample $\ket{\Psi} \leftarrow \textnormal{Gen}$. Let $s' := \textnormal{Ver}(\ket{\Psi})$.
    \item Send $(\textsf{Trigger},ssid, (\textsf{ChallengeClaim}, x, s'), 0)$ to $\mathcal{F}_{Ledg}$. (If $P$ was honest, this updates $ssid.\textsf{state} \leftarrow (s', \text{``No active claim''})$ and releases $d_0$ coins to $P$).
\end{itemize} 
\end{itemize}

\rule[2ex]{16.5cm}{0.5pt}\vspace{-.5cm}
\caption{Protocol for preventing malicious attempts at changing the state variable $\textsf{serial}$ of a smart contract. }
  \label{fig: protocol preventing}
  
\end{figure}

\paragraph{Trading a quantum banknote for coins:} Finally, we describe a protocol for trading a quantum banknote to recover all the coins deposited in its associated banknote-contract. 

\begin{figure}[H]
\rule[1ex]{16.5cm}{0.5pt}

Protocol carried out by a party $P$.\\

Input of $P$: A tuple $(\ket{\Psi}, s, ssid, d)$, where $\ket{\Psi}$ is a quantum banknote with serial number $s$, and $ssid$ is the session identifier of a banknote-contract such that $ssid.\textsf{state} = (z, \text{``No active claim''})$ and $ssid.\textsf{coins} = d$.\\ 

\begin{itemize}
\item Run $\mathcal{A_*}$ on $\ket{\Psi}$ to get outcome $x \in \{0,1\}^{\lambda}$ such that $H(x) = s$. 
\item Send message $(\textsf{Trigger}, ssid, (\textsf{RecoverCoins}, x), 0)$ to $\mathcal{F}_{Ledg}$. This releases $d$ coins to $P$.
\end{itemize}

\rule[2ex]{16.5cm}{0.5pt}\vspace{-.5cm}
\caption{Protocol for trading a quantum banknote for coins.}
  \label{fig: protocol trading banknotes}
  
\end{figure}

\section{Security}
\label{sec: security}

We first specify an adversarial model. Security with respect to this adversarial model is formally captured by Theorem \ref{thm: security}. At a high-level, Theorem \ref{thm: security} establishes that, within this adversarial model, no adversary can increase his ``value'' beyond what he has legitimately spent or received to and from honest parties. This captures, for example, the fact that the adversary will not be able to double-spend his banknotes, or successfully file a ``lost banknote claim'' for banknotes he does not legitimately possess.

\paragraph{Adversarial model} 

We assume that all the messages of honest parties are sent using the ideal functionality for authenticated communication $\mathcal{F}_{auth}$, and that the adversary sees all messages that are sent (in UC language, we assume that the adversary is activated every time a party sends a message) but has no control over the delivery of messages (whether they are delivered or not) and their timing. Our payment system can be made to work also if we assume that the adversary can delay delivery of honest parties' messages by a fixed amount of time (see the remark preceding Fig. \ref{fig: protocol making a claim} for more details), but, for simplicity, we do not grant the adversary this power.

The adversary can corrupt any number of parties, and it may do so adaptively, meaning that the corrupted parties are not fixed at the start, but rather an honest party can become corrupted, or a corrupted party can return honest, at any point. The process of corruption is modeled analogously as in the original UC framework, where the adversary simply writes a \textsf{corrupt} message on the incoming tape of an honest party, upon which the honest party hands all of its information to the adversary, who can send messages on the corrupted party's behalf. Our setting is slightly more involved in that corrupted parties also possess some quantum information, in particular the quantum banknotes. We assume that when an adversary corrupts a party he takes all of its quantum banknotes. Importantly, we assume that these are not returned to the party once the party is no longer corrupted. It might seem surprising that we do not upper bound the fraction of corrupted parties. Indeed, such a bound would only be needed in order to realize securely the ideal functionality $\mathcal{F}_{Ledg}$ (any consensus-based realization of $F_{Ledg}$ would require such a bound). Here, we assume access to such an ideal functionality, and we do not worry about its secure realization. Naturally, when replacing the ideal functionalities with real-world realizations one would set the appropriate bound on the corruption power of the adversary, but we emphasize that our schemes are independent of the particular real-world realization. Note that we do not fix a set of parties at the start, but rather new parties can be created (see below for more details).


We assume that (ITMs of) honest parties run the code $\pi$. This represents the ``honest'' code which executes the protocols from section \ref{sec: main} as specified. The input to $\pi$ then specifies when and which protocols from section \ref{sec: main} are to be executed. As part of $\pi$, we specify that, upon invocation, a party sends a message \textsf{Register} to $\mathcal{F}_{Ledg}$ to register itself. We also specify as part of $\pi$ that an honest party runs the protocol of Fig. \ref{fig: protocol preventing} (to prevent malicious claims for lost banknotes). Moreover, for notational convenience, we specify as part of $\pi$ that each party maintains a local variable $\textsf{banknoteValue}$, which keeps track of the total value of the quantum banknotes possessed by the party. $\textsf{banknoteValue}$ is initialized to $0$, and updated as follows. Whenever a party $P$ successfully receives a quantum banknote (i.e. $P$ is the payee in the protocol from Fig. \ref{fig: protocol making a payment} and does not abort) of value $d$ (i.e. the associated smart contract has $d$ coins deposited), then $P$ updates $\textsf{banknoteValue} \leftarrow \textsf{banknoteValue} +d$. Similarly, when $P$ sends a quantum banknote of value $d$, it updates $\textsf{banknoteValue} \leftarrow \textsf{banknoteValue} -d$.
Finally, we specify also as part of $\pi$, that whenever a party that was corrupted is no longer corrupted, it resets $\textsf{banknoteValue} = 0$ (this is because we assumed that quantum banknotes are not returned by the adversary). The following paragraph leads up to a notion of security and a security theorem.

Let $\mathcal{A}$ be a quantum polynomial-time adversary and $\mathcal{E}$ a quantum polynomial-time environment. Consider an execution of $\pi$ with adversary $\mathcal{A}$ and environment $\mathcal{E}$ (see section \ref{sec: uc} for more details on what an ``execution'' is precisely). We keep track of two quantities during the execution, which we denote as \textsf{AdversaryValueReceived} and \textsf{AdversaryValueCurrentOrSpent} (These quantities are not computed by any of the parties, adversary or environment. Rather, they are just introduced for the purpose of defining security). The former represents the amount of value, coins or banknotes, that the adversary has received either by virtue of having corrupted a party, or by having received a payment from an honest party. The latter counts the total number of coins currently possessed by corrupted parties, as recorded on $\mathcal{F}_{Ledg}$, and the total amount spent by the adversary to honest parties either via coins or via quantum banknotes (it does not count the value of quantum banknotes currently possessed; these only count once they are successfully spent). Both quantities are initialized to $0$, and updated as follows throughout the execution:
\begin{itemize}
    \item[(i)] When $\mathcal{A}$ corrupts a party $P$: let $d$ be the number of coins of $P$ according to the global functionality $\mathcal{F}_{Ledg}$ and $d'$ be $P$'s \textsf{banknoteValue} just before being corrupted. Then, $\textsf{AdversaryValueReceived} \leftarrow \textsf{AdversaryValueReceived} + d + d'$, and $\textsf{AdversaryValueCurrentOrSpent} \leftarrow \textsf{AdversaryValueCurrentOrSpent} + d$.
    \item[(ii)] When a corrupted party $P$ with $d$ coins and $\textsf{banknoteValue} = d'$ ceases to be corrupted and returns honest, $\textsf{AdversaryValueReceived} \leftarrow \textsf{AdversaryValueReceived} - d$.
    \item[(iii)] When an honest party pays $d$ coins to a corrupted party, $\textsf{AdversaryValueReceived} \leftarrow \textsf{AdversaryValueReceived} + d$. Likewise, when an honest party sends a quantum banknote of value $d$ to a corrupted party, through the protocol of Fig. \ref{fig: protocol making a payment}, then (even if the corrupted party does not return \textsf{accept}) 
    $\textsf{AdversaryValueReceived} \leftarrow \textsf{AdversaryValueReceived} + d$.
    \item[(iv)] When $\mathcal{A}$ succesfully spends a quantum banknote of value $d$ to an honest party $P$, i.e. a corrupted party is the payer in the protocol from Fig. \ref{fig: protocol making a payment} and $P$ is the payee and returns \textsf{accept}, or when $\mathcal{A}$ pays $d$ coins to an honest party, then $\textsf{AdversaryValueCurrentOrSpent} \leftarrow \textsf{AdversaryValueCurrentOrSpent} + d$.
    \item[(v)]
    When a corrupted party receives $d$ coins from a banknote-contract, then $\textsf{AdversaryValueCurrentOrSpent} \leftarrow \textsf{AdversaryValueCurrentOrSpent} + d$. Notice that this can happen only in two ways: $\mathcal{A}$ successfully converts a quantum banknote of value $d$ to coins on $\mathcal{F}_{Ledg}$ (via the protocol of Fig. \ref{fig: protocol trading banknotes}), or a corrupted party successfully challenges a \textsf{BanknoteLost} claim (in this case $d = d_0$).
\end{itemize}
Intuitively, if our payment scheme is secure, then at no point in time should the adversary be able to make $\textsf{AdversaryValueCurrentOrSpent} - \textsf{AdversaryValueReceived} > 0$. This would mean that he has successfully spent/stolen value other than the one he received by virtue of corrupting a party or receiving honest payments. The following theorem formally captures this notion of security. First, we denote by $\mathcal{F}_{Ledg}\text{-}\textrm{EXEC}^{(MaxNetValue)}_{\pi, \mathcal{A}, \mathcal{E}}(\lambda, z)$ the maximum value of $\textsf{AdversaryValueCurrentOrSpent} - \textsf{AdversaryValueReceived}$ during an execution of $\pi$ with adversary $\mathcal{A}$ and environment $\mathcal{E}$, with global shared functionality $\mathcal{F}_{Ledg}$.

\begin{theorem}[Security]
\label{thm: security}
For any quantum polynomial-time adversary $\mathcal{A}$ and quantum polynomial-time environment $\mathcal{E}$, 
$$ \Pr[\mathcal{F}_{Ledg}\text{-}\textrm{EXEC}^{(MaxNetValue)}_{\pi, \mathcal{A}, \mathcal{E}}(\lambda, z) > 0] = negl(\lambda).$$
\end{theorem}

The rationale behind considering executions of $\pi$ and quantifying over all possible adversaries and environments is that doing so captures all possible ways in which a (dynamically changing) system of honest parties running our payment system alongside an adversary  can behave (where the adversary respects our adversarial model).

Recall that, in an execution of $\pi$, the environment has the ability to invoke new parties and assign to them new unique PIDs. Since in $\mathcal{F}_{Ledg}$ the PIDs are used to register parties and initialize their number of coins, this means that the environment has the ability to pick the initial number of coins of any new party that it invokes. Moreover, by writing inputs to the parties input tapes, the environment can instruct honest parties to perform the honest protocols from section \ref{sec: main} in any order it likes. Quantifying over all adversaries and environments, in the statement of Theorem \ref{thm: security} means that the adversary and the environment can intuitively be thought of as one single adversary. The statement of the theorem thus captures security against realistic scenarios in which new parties can be adversarially created with an adversarially chosen number of coins, and they can be instructed to perform the honest protocols of the payment system from section \ref{sec: main}, in whatever sequence is convenient to the adversary. 

\begin{proof}[Proof of Theorem \ref{thm: security}]
The proof is straightforward, and we avoid being overly formal.

Suppose for a contradiction that there exists $\mathcal{A}$ and $\mathcal{E}$ such that
\begin{equation}
\label{eq: 1}
\Pr[\mathcal{F}_{Ledg}\text{-}\textrm{EXEC}^{(MaxNetValue)}_{\pi, \mathcal{A}, \mathcal{E}}(\lambda, z) > 0] \neq negl(\lambda).
\end{equation}

Then, we go through all of the possible ways that an adversary can increase its net value, i.e. increase the quantity $\textsf{AdversaryValueCurrentOrSpent} - \textsf{AdversaryValueReceived}$: the adversary can do so through actions from items (ii), (iv) and (v) above. Amongst these, it is easy to see that action (ii) never results in $\textsf{AdversaryValueCurrentOrSpent} - \textsf{AdversaryValueReceived} > 0$. Thus, in order for \eqref{eq: 1} to hold, it must be the case that one of the following happens with non-negligible probability within an execution of $\pi$ with adversary $\mathcal{A}$ and environment $\mathcal{E}$.
\begin{itemize}
\item An action from item (iv) resulted in a positive net value for $\mathcal{A}$, i.e. $\textsf{AdversaryValueCurrentOrSpent} - \textsf{AdversaryValueReceived} > 0$. Notice that for this to happen it must be the case that $\mathcal{A}$ has double-spent a banknote, i.e. $\mathcal{A}$ has produced two banknotes with the same serial number that have both been accepted by honest parties in a payment protocol of Fig. \ref{fig: protocol making a payment}, and so they have both passed verification. But then, it is straightforward to see that we can use this adversary, together with $\mathcal{E}$ to construct an adversary $\mathcal{A}'$ that breaks the security of the quantum lightning scheme: $\mathcal{A}'$ simply simulates an execution of protocol $\pi$ with adversary $\mathcal{A}$ and environment $\mathcal{E}$, and with non-negligible probability the adversary $\mathcal{A}$ in this execution produces two banknotes with the same serial number. $\mathcal{A}'$ uses these banknotes to win the security game of quantum lightning. 
\item An action from item (v) resulted in a positive net value for $\mathcal{A}$. Then, notice that for this to happen it must be that either: 

\begin{itemize}
    \item $\mathcal{A}$ has sent a message $(\textsf{Trigger}, ssid, (\textsf{RecoverCoins}, x), 0)$ to $\mathcal{F}_{Ledg}$ for some $ssid$ and $x$ such that $H(x) =s$, where $ssid.\textsf{state} = (s, \text{``No active claim''})$, and the last ``make a payment'' protocol (from Fig. \ref{fig: protocol making a payment}) referencing $ssid$ had an honest party as payee which remained honest at least up until after $\mathcal{A}$ sent his message (or the banknote-contract was initialized by an honest user and the banknote was never spent). But then, one of the following must have happened: 
\begin{itemize}
\item $\mathcal{A}$ possessed a bolt $\ket{\Psi}$ with serial number $s$ at some point, before $\ket{\Psi}$ was spent to the honest user. Then, this adversary would have recovered a good $x$ and also spent a bolt with serial number $s$ successfully to an honest user. But such an $\mathcal{A}$, together with $\mathcal{E}$, can be used to win $\textsf{Game}_{2}$ from Definition \ref{def: extra property} with non-negligible probability, with a similar reduction to the one above, thus violating the property of Definition \ref{def: extra property}. 
\item $\mathcal{A}$ recovered $x$ such that $H(x) = s$ without ever possessing a valid bolt with serial number $s$. Again, such an adversary could be used, together with $\mathcal{E}$ to win $\textsf{Game}_{2}$ from Definition \ref{def: extra property}).
\item $\mathcal{A}$ has successfully changed the serial number of contract $ssid$ to $s$ from some previous $s'$ without possessing a bolt $\ket{\Psi}$ with serial number $s'$. This cannot happen since any honest user who possesses the valid bolt with serial number $s'$ performs the protocol of Fig. \ref{fig: protocol preventing}.
\end{itemize}

\item $\mathcal{A}$ has sent a message $(\textsf{Trigger}, ssid, (\textsf{ChallengeClaim}, x), 0)$ to $\mathcal{F}_{Ledg}$ for some $ssid$ such $ssid.\textsf{state} = (s,\text{``Claim by $pid$ at time $t$''})$ for some $s, pid, t$ with $pid$ honest and $x$ such that $H(x) =s$. Since $pid$ is honest, he must be the last to have possessed a valid bolt with serial number $s$. Then, there are two possibilities:
\begin{itemize}
\item $\mathcal{A}$ never possessed a valid bolt with serial number $s$, and succeeded in recovering $x$ such that $H(x) =s$. Analogously to earlier, this adversary, together with $\mathcal{E}$, can be used to win $\textsf{Game}_{2}$.
\item $\mathcal{A}$ possessed a bolt $\ket{\Psi}$ with serial number $s$ at some point, before $\ket{\Psi}$ was spent to an honest user. Analogously to earlier, this means such an $\mathcal{A}$ both recovered an $x$ with $H(x) = s$ and spent a bolt with serial number $s$ successfully. Such an $\mathcal{A}$ can be used, together with $\mathcal{E}$, to win $\textsf{Game}_{2}$. 
\end{itemize}

\end{itemize} 

\end{itemize}

\end{proof}

\paragraph{Other attacks that do not allow the adversary to gain value, but disrupt honest parties} 
We briefly comment on the possibility of attacks which do not make the adversary profit, but may disrupt honest parties. The security definition for quantum lightning, or the security proved in Theorem \ref{thm: security}, do not directly exclude a scenario in which an adversary is able to dishonestly generate a quantum banknote, spend it succesfully with non-negligible probability, but when the honest party who received it tries to spend it at a later time verification will fail, or she will fail to obtain a valid classical certificate from it. Such issues ought to be considered in general, but fortunately they do not arise when the underlying quantum lightning scheme has a verification procedure which is a rank-1 projection, since any banknote that passes verification is automatically projected onto an honestly generated banknote. This is the case for the candidate construction based on the multi-collision resistance of certain degree-2 hash functions of \cite{zhandry2017quantum}.









\section{Conclusion}
In this work, we gave the first example of the use of classical smart contracts in conjunction with quantum cryptographic tools. We showed that smart contracts can be combined with quantum tools, in particular quantum lightning, to design a decentralized payment system which solves the problem of scalability of (payment) transactions. There is currently only one known secure construction of quantum lightning, which relies on a computational assumption about multi-collision resistance of certain degree-2 hash functions \cite{zhandry2017quantum}. Finding alternative constructions of quantum lightning, secure under more well-studied computational assumptions, is a very interesting open problem.

Smart contracts have found several applications in classical cryptographic tasks, but their application to quantum cryptographic tasks is virtually unexplored. We hope that this work will ignite future investigations. Some candidate tasks which might potentially benefit from smart contracts are: generation of public trusted randomness, distributed delegation of quantum computation, secure multi-party quantum computation. 

\section*{Acknowledgements}
The author thanks Andru Gheorghiu, Abel Molina, Mario Larangeira, Or Sattath and Thomas Vidick for valuable comments and discussions on earlier versions of this work. The author especially thanks Or Sattath for pointing out an imprecision in the previous version of Definition \ref{def: extra property}. The author also thanks Ran Canetti, Ranjit Kumaresan and Dominique Unruh for helpful email exchanges. The author is supported by the Kortschak Scholars program and AFOSR YIP award number FA9550-16-1-0495.

\bibliographystyle{alpha}
\bibliography{references}

\appendix
\section{Appendix}
\subsection{Proof of Proposition \ref{prop: extra property 2}}
\label{sec: appendix}
\begin{proof}[Proof of Proposition \ref{prop: extra property 2}]
We assume some familiarity with Zhandry's construction (see section 6 of \cite{zhandry2017quantum} for more details). In his construction, a full bolt is a tensor product of $n$ mini-bolts. A valid mini-bolt with serial number $y$ takes the form $\ket{\Psi}^{\otimes (k+1)}$ where $\ket{\Psi}$ is a superposition of pre-images of $y$ under a certain function $f$ (this is called $f_{\mathcal{A}}$ in Zhandry's paper, and we do not go into the details of what this function is). The serial number of the full bolt is the concatenation of the serial numbers of the mini-bolts. The computational assumption under which Zhandry's construction is proved secure is that $f$ is $(2k+2)$-multi-collision resistant, i.e. it is hard to find $2k+2$ colliding inputs (for this particular function it is easy to find $k+1$ on the other hand). Let $\tilde{H}$ be the function that takes as input $k+1$ pre-images $z_{1},..,z_{k+1}$ and outputs $f(z_i)$ if this value is the same for all $i$, and outputs $\perp$ otherwise. Take $H := \tilde{H}^{\times n}$. 

For property $(I)$, $\mathcal{A}_*$ again just measures all registers in the computational basis. For property $(II)$, similarly to the proof of Proposition \ref{prop: extra property 1}, we can construct an adversary $\mathcal{A}'$ that breaks the $(2k+2)$-multi-collision resistance of $f$ from an adversary $\mathcal{A}$ that wins $\textsf{Game}_{12}$ with non-negligible probability: $\mathcal{A}'$ runs $(\textnormal{Gen}, \textnormal{Ver}, H) \leftarrow \textsf{QL.Setup}(1^\lambda)$ and samples $\ket{\psi}\leftarrow \textnormal{Gen}$. $\mathcal{A}'$ gives $(\textnormal{Gen}, \textnormal{Ver}, H, \ket{\psi})$ as input to $\mathcal{A}$. $\mathcal{A}$ returns $x$ which is parsed as $(x_1,..,x_n)$, where each $x_i$ is a $(k+1)$-tuple, and $\ket{\psi'}$. $\mathcal{A}'$ then measures each of the $n$ registers of $\ket{\psi'}$ to get $n$ $(k+1)$-tuples $ (x'_1,..,x'_n)$. If there is some $i$ such that $(x_i, x_i')$ is a $(2k+2)$-collision, then $\mathcal{A}'$ outputs this. With non-neglibile probability $\mathcal{A}'$ outputs a $(2k+2)$-collision: in fact, since $\mathcal{A}$ wins $\textsf{Game}_{12}$ with non-negligible probability, then $\ket{\psi'}$ must pass Ver with non-negligible probability; from the analysis of Zhandry's proof, we know that any full bolt that passes verification with non-negligible probability must be such that most mini-bolts have non-negligible weight on most pre-images. Since $k$ can be taken to be constant in the size of the security parameter (it is even possible to take $k=1$), then, for any $i$, there is a non-negligible probability that the $k+1$ entries of $x_i$ are all distinct from the entries of $x_i'$, which gives a $(2k+2)$-collision.
\end{proof}

\end{document}